\documentclass[runningheads]{llncs}
\usepackage{dsfont}
\usepackage{amsfonts}
\usepackage{amsmath}
\usepackage{amssymb}
\usepackage{pifont}

\usepackage{booktabs}
\usepackage{multicol}
\usepackage{multirow}
\usepackage{tabularx}
\usepackage{diagbox}
\usepackage{graphicx}

\usepackage[T1]{fontenc}
\usepackage{graphicx}
\usepackage{color}
\usepackage{tikz}
\usetikzlibrary{calc}
\usepackage[ruled,linesnumbered]{algorithm2e}
\usepackage{algpseudocode}
\usepackage{bbm}
\usepackage{subcaption}
\usepackage{xcolor}
\usepackage{array}
\usepackage{fancyhdr}
\usepackage{hyperref}
 \hypersetup{
     colorlinks=true,
     linkcolor=blue,
     filecolor=blue,
     citecolor = black,      
     urlcolor=cyan,
     }
\usepackage{cleveref}
\usepackage{xcolor}
\usepackage{tikz}
\usepackage{verbatim}
\usepackage{listings}
\definecolor{mygreen}{RGB}{28,172,0} 
\definecolor{mylilas}{RGB}{170,55,241}

\usepackage{matlab-prettifier}
\usepackage{color}
\usepackage[shortlabels]{enumitem}

\usepackage{color, colortbl}
\usepackage{tabularx}
\definecolor{Gray}{gray}{0.9}
\usepackage{enumitem}
\usepackage[font={small,it},skip=3pt]{caption}
\usepackage{appendix}

\usepackage{graphicx}
\usepackage{subcaption}

\let\oldnl\nl
\newcommand{\nonl}{\renewcommand{\nl}{\let\nl\oldnl}}

\begin{document}
\title{ MEGA-PT: A Meta-Game Framework for Agile Penetration Testing }
\author{Yunfei Ge\inst{1} \and
Quanyan Zhu\inst{1} \vskip -3mm}
\authorrunning{Y. Ge et al.}
\institute{New York University, New York NY 11201, USA \\
\email{\{yg2047,qz494\}@nyu.edu}}
\maketitle            
\begin{abstract}\vskip -8mm
Penetration testing is an essential means of proactive defense in the face of escalating cybersecurity incidents. Traditional manual penetration testing methods are time-consuming, resource-intensive, and prone to human errors. Current trends in automated penetration testing are also impractical, facing significant challenges such as the curse of dimensionality, scalability issues, and lack of adaptability to network changes. To address these issues, we propose MEGA-PT, a meta-game penetration testing framework, featuring micro tactic games for node-level local interactions and a macro strategy process for network-wide attack chains. The micro- and macro-level modeling enables distributed, adaptive, collaborative, and fast penetration testing.  MEGA-PT offers agile solutions for various security schemes, including optimal local penetration plans, purple teaming solutions, and risk assessment, providing fundamental principles to guide future automated penetration testing. Our experiments demonstrate the effectiveness and agility of our model by providing improved defense strategies and adaptability to changes at both local and network levels.

\end{abstract}

\section{Introduction}

With the exponential growth of network technologies and the escalating frequency of security incidents, cybersecurity has become a global concern \cite{ge2023gazeta,zhao2021combating}. In response to these challenges, penetration testing has emerged as a crucial solution for uncovering system vulnerabilities and assessing network security through authorized ethical attacks \cite{hu2020automated}. However, traditional manual penetration testing performed by skilled IT professionals has several limitations. It can be time-consuming, resource-intensive, and prone to human error. Relying solely on manual testing often falls short of identifying all vulnerabilities within the system. Thus, there is a need for automation and the integration of advanced threat intelligence into the penetration testing process, enabling a more efficient and scalable approach to enhancing cybersecurity.

Current proposed automated penetration testing methods are increasingly becoming non-standard, complex, and resource-consuming, despite tool advancements. Reinforcement learning (RL) or Markov Decision Process (MDP) based methods \cite{ghanem2019reinforcement,hu2020automated} suffer from the curse of dimensionality, as they define the state space as the collection of all known information for each machine on the network. Partially Observable Markov Decision Process (POMDP) methods \cite{shmaryahu2017partially} face scalability issues, making it unfeasible to model and solve for large networks. Additionally, these methods lack adaptability to changes, as they assume the network structure and software configuration remain unchanged to learn the optimal policy. Many proposed models do not follow the Tactics Techniques and Procedures (TTPs) in real cybersecurity practice, relying mainly on hypotheses and simulations, which undermines their transition to praxis. Furthermore, merely identifying vulnerabilities through penetration testing is insufficient; it is crucial to provide defense suggestions and risk analysis based on the testing to enhance overall security.

To address the limitations of current penetration testing methods, we propose a meta-game-based automated penetration testing framework (MEGA-PT). In this framework, the micro tactic game captures the interactions between the defender and attacker at each local node, while the macro strategy process models lateral movement and the attack chain across the entire network. This approach offers several key features: practical implications, as the sequential interactions in each micro tactic game follow the MITRE ATT\&CK framework \cite{mitre2020mitigations} and use extensive-form games to model attack/defense dynamics; distributed penetration testing, with modularized processes at each micro tactic game allowing for parallel computation; and adaptability to changes at both the local and network levels, ensuring efficient testing and scalability.

Our proposed model enables various security schemes, depending on the solution concept selected for the meta-game. This extension of penetration testing goes beyond vulnerability discovery to include defense strategy recommendations and risk analysis. Specifically, the model provides solutions for the following security schemes: optimal local penetration plans under certain defense strategies, purple teaming solutions for enhanced defense suggestions, and risk assessments at equilibrium. 
Our contributions can be summarized as follows:
\begin{enumerate}
    \item We propose a meta-security game framework MEGA-PT for automated penetration testing, where micro tactic games at each local node are modeled as extensive-form games, and the macro strategy process is modeled as a Markov decision process.
    \item We offer applicable solution concepts for security schemes aimed at vulnerability discovery, defense suggestion, and risk analysis.
    \item Our experiments demonstrate the effectiveness of MEGA-PT by providing improved defense strategies and adaptability to changes at both local and network levels.
    \item In essence, MEGA-PT establishes fundamental principles to drive the future of automated penetration testing and its practices.
\end{enumerate}

\section{Problem Formulation}

Penetration testing is an ethical attack aimed at identifying system vulnerabilities, providing defense suggestions, and offering risk assessments. In this context, the term \textit{attacker} refers to the penetration testing agent, while the \textit{defender} represents the system security management engine. To describe attacker behaviors within a security program, Tactics, Techniques, and Procedures (TTPs) are commonly used. Figure~\ref{fig:ttps} illustrates the hierarchy between these terms. For security strategy analysis, we focus on Tactics and Techniques in penetration testing, omitting the detailed Procedures.

\begin{figure}[!t]
    \centering
    \includegraphics[width=0.9\linewidth]{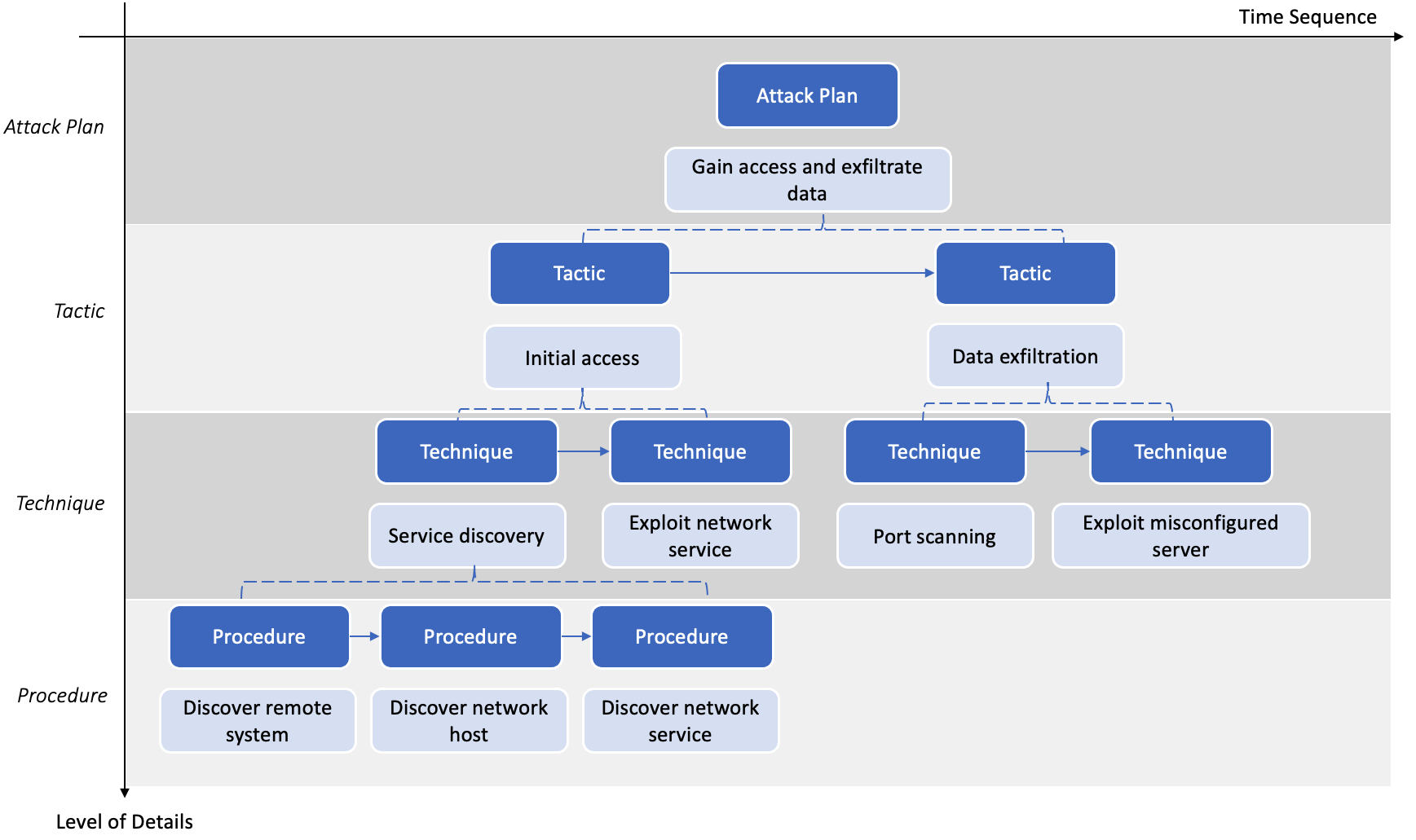}
    \caption{Attack plan and tactics, techniques, and procedures (TTPs). Depending on the level of detail, each attack plan can be elaborated by a sequence of tactics (from left to right), where each tactic is composed of a sequence of techniques, and each technique can be described by a sequence of procedures.}\vskip -5mm
    \label{fig:ttps}
\end{figure}

To describe interactions in penetration testing using TTPs, we propose a meta-security game over a network graph. The macro strategic game represents strategic attack activities between nodes, while the micro tactic game details tactic-level attack procedures on each local node. Let the directed graph \(G = \langle \mathcal{V}, \mathcal{E} \rangle\) represent the target network topology, where \(\mathcal{V}\) is a set of nodes (e.g., server, database, device), and \(\mathcal{E} \subseteq \mathcal{V} \times \mathcal{V}\) is a set of directed edges representing connections (e.g., SSH, RDP, cloud services) from node \(u\) to node \(v\). Self-loops are allowed as they indicate continued exploration of the same node. Let \(v^0 \in \mathcal{V}\) be the initial foothold in the system. Figure~\ref{fig:example} shows an example of the networked system topology. The penetration tester, as an ethical attacker, aims to explore available information, exploit discovered vulnerabilities, and influence critical assets inside the network.

\begin{figure}[!t]
\centering
\includegraphics[width=0.5\linewidth]{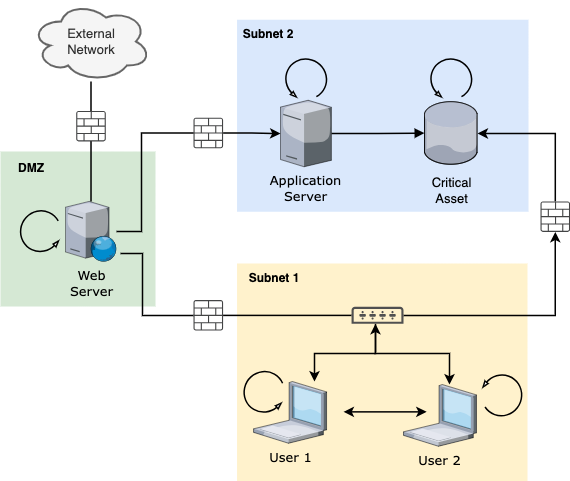}
\caption{Illustration of the networked system topology. The system contains $5$ nodes (web server, application server, $2$ user devices, and critical asset). The penetration testing starts from the web server, which is open to the external network.}\vskip -5mm
\label{fig:example}
\end{figure}

\subsection{Micro Tactic Games}

\subsubsection{Game-theoretic Modeling}\hfill

When an attacker gains access to one node, there are multiple steps involved before he completes the exploration and exploitation process on the node.
To model the sequential moves at each local node, we use the concept of an extensive-form game tree to explicitly and visually represent the sequential moves, possible outcomes, and information available at each decision point in a strategic interaction. Figure~\ref{fig:extensive} illustrates an example of a game tree at the web server. The attacker can choose to perform reconnaissance on hosts and services on the web server and then exploit the host to perform privilege escalation. The defender can choose to accept or deny the access request based on their defense policy. Depending on the privilege levels, the attacker could collect different credentials in the game, leading to various expected tactic outcomes and connecting to different nodes in the network.

\begin{definition}[Micro Tactic Game (MTG)]\label{def:mtg}
The Micro Game of the MEGA-PT is defined by a set of Micro Tactic Games (MTG) $\{\Gamma^v\}_{v\in\mathcal{V}}$ where $\mathcal{V}$ is the set of nodes in the system. Given a node $v\in\mathcal{V}$ in the network, the MTG on node $v$ can be represented by an extensive-form game tuple $\Gamma^v = \langle \mathcal{N}\cup \{c\}, \mathcal{H}^v,P, \{\mathcal{A}^v_i\}_{i\in \mathcal{N}\cup c}, \sigma^v_c,  \{u^v_i\}_{i\in \mathcal{N}}, \mathcal{Z}^v\rangle$, where each components represents:
\begin{itemize}
    \item \textbf{Players $\mathcal{N}=\{a,d\}$}  There are two main players in the game: the attacker ($a$) and the defender ($d$). Additionally, $c$ is the nature that represents the system randomness.
    \item \textbf{Histories $\mathcal{H}^v$} Each vertex in the game tree $h\in \mathcal{H}^v$ corresponds to a unique sequence of actions taken from the beginning of the game.
    \item \textbf{Turn Function $P: \mathcal{H}^v\mapsto \mathcal{N}\cup \{c\}$} The function $P(h)$ determines whose turn it is to make a move at each decision point for a given history vertex $h$. 
    \item \textbf{Techniques $\mathcal{A}^v_i$} $\mathcal{A}^v_i$ is a set of techniques that player $i$ can take. $A(h)$ denote the feasible techniques for player $i = P(h)$ at vertex $h\in\mathcal{H}^v$.
    \item \textbf{System Randomness $\sigma^v_c\in\Sigma^v_c$} Nature’s fixed policy $\sigma^v_c$ specifies the system randomness, which could be related to network traffic load, randomized system configuration, hardware failures, etc.
    \item \textbf{Tactic Expected Outcomes $\mathcal{Z}^v$} $\mathcal{Z}^v$ represents the finite set of possible outcomes for each attack sequence in the MTG. These outcomes correspond to the results observed at the leaf vertices of the game tree, which could be the credentials to user devices, authorized connection to the server, no vulnerability found, etc.
    \item \textbf{Utilities $u^v_i:\mathcal{Z}^v\mapsto \mathds{R}$}  The utility function $u^v_i$ determines the payoff or cost player $i$ receives when reaching a certain outcome.
\end{itemize}
\end{definition}

\begin{figure}
\centering
\includegraphics[width=0.98\linewidth]{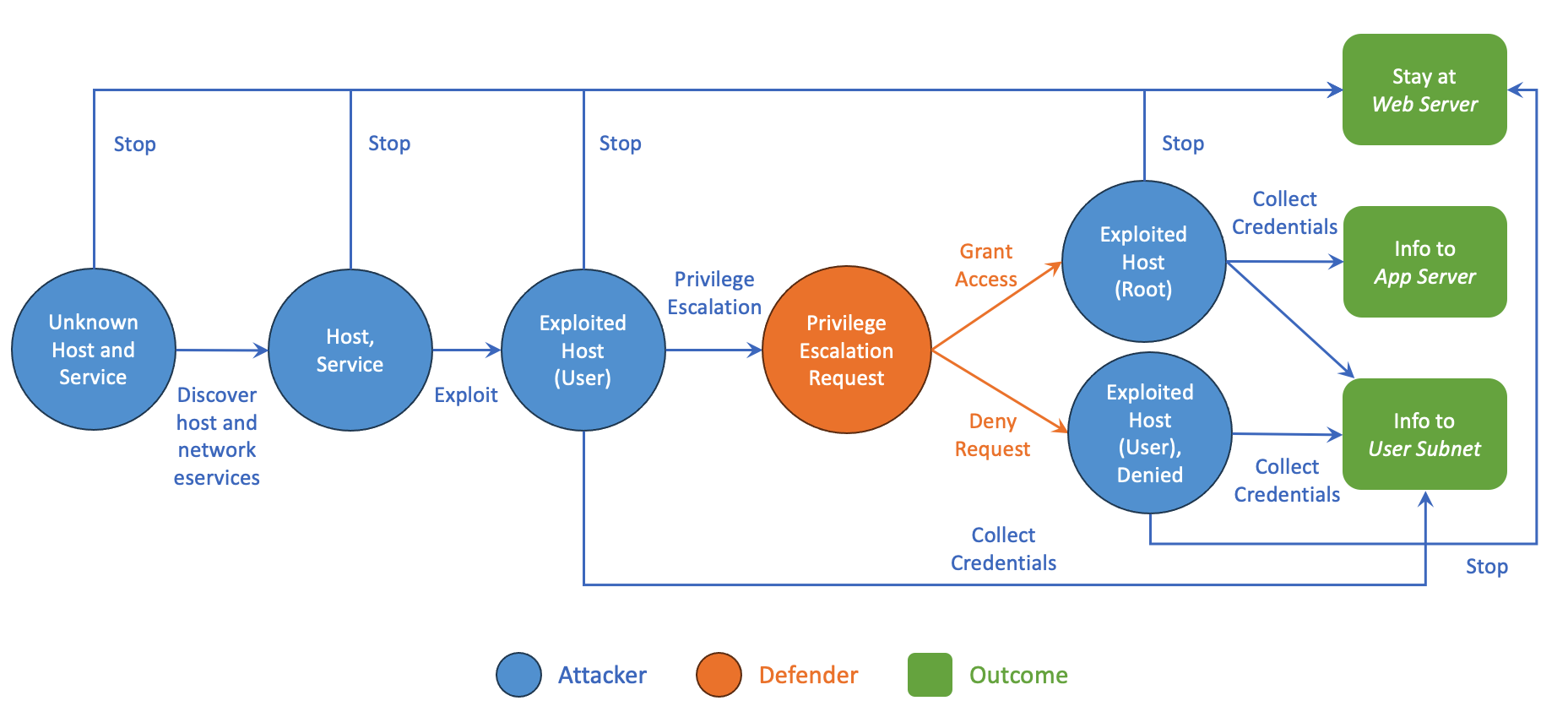}
\caption{Micro Tactic Game at the web server. The attacker needs to discover the host and services on the node, requesting privilege escalation to collect the credentials leading to other nodes. The defender could grant or deny the attacker's request depending on the defense strategy. The players' sequence of actions would lead to different expected tactic outcomes.}\vskip -5mm
\label{fig:extensive}
\end{figure}

In the context of TTPs, the attack tactic at the current node corresponds to a sequence of techniques, while the outcomes represent the high-level tactical goals. In this work, we assume that the tactical outcomes are either staying in the current node or leading to another node that can be connected from the current node. Thus, with a slight abuse of notation, we denote $\mathcal{Z}^v = \{u \mid u \in \mathcal{V}, (v,u) \in \mathcal{E}\}$.

\subsubsection{Penetration Plans}\hfill

Before we define tactics or strategies, we need to understand the basis on which players make their decisions. For any strategic player, decisions are made based on the current knowledge of the situation. However, it is sometimes challenging for the player to obtain the complete interaction history due to partial observations. 
Consequently, there are decision vertices in the game tree that the player cannot distinguish between. In an extensive-form game, this is called an \textit{information set}. In this work, we refer to this information set as a \textit{knowledge set}.
 
\begin{definition}[Knowledge Set]
Given the MTG at node $v\in\mathcal{V}$, a knowledge set $I_i\subseteq \mathcal{H}^v$ of a player $i$ represents a set of decision vertices where the player $i$ has the same available techniques and cannot distinguish between the vertices.
\end{definition}

The concept of knowledge set helps us better describe the decision-making process for both players. Given the MTG at the current node, attackers can construct their tactics in different ways. One approach is to have a sequential plan of techniques from the beginning to the end of the game. This step-by-step pure penetration or defense plan assigns a single technique to each possible knowledge set. Players can also randomize over single-technique plans at each knowledge set, known as a mixed penetration or defense plan. 

\begin{definition}[Pure Penetration (Defense) Plan]
Consider the MTG defined in Definition~\ref{def:mtg}. 
Given the MTG $\Gamma^v$, a pure penetration or defense plan at node $v\in\mathcal{V}$ for player $i\in \mathcal{N}$  is a mapping $q^v_i: \mathcal{I}_i \mapsto \mathcal{A}^v_i$ that assigns a technique $q^v_i(I_i)\in A(I_i)$ for every knowledge set $I_i \in \mathcal{I}_i$. Denote $Q^v_i$ as the set of all possible pure penetration or defense plans for player $i \in \mathcal{N}$ at this micro game.
The pure penetration or defense plan for the entire system is defined as the set $\{q^v_i\}_{v\in\mathcal{V}}$, with $i=a$ means the attacker and $i=d$ means the defender.
\end{definition}

\begin{definition}[Mixed Penetration (Defense) Plan]
Consider the MTG defined in Definition~\ref{def:mtg}. 
Given the MTG $\Gamma^v$, a mixed penetration or defense plan at node $v\in\mathcal{V}$ for player $i\in \mathcal{N}$ is a probability distribution over all of player $i$'s pure penetration plans, i.e., $\sigma^v_i \in \Delta(Q^v_i)$. Denote $\Sigma^v_i$ as the set of all possible mixed penetration or defense plans for player $i \in \mathcal{N}$ at this micro game.
The mixed penetration or defense plan for the entire system is defined as the set $\{\sigma^v_i\}_{v\in\mathcal{V}}$, with $i=a$ means the attacker and $i=d$ means the defender.
\end{definition}

The other approach is to focus on each knowledge set instead of defining the step-by-step actions of the entire game. 
At each knowledge set, a probabilistic distribution is assigned over the feasible techniques. This corresponds to the behavioral strategy in extensive-form games. Instead of planning everything ahead, this policy focuses on the decisions in each knowledge set. In this work, we call it an operational search plan. Denote by $\mathcal{I}_i$ the collection of knowledge sets of player $i \in \mathcal{N}$. By definition, for every knowledge set $I_i \in \mathcal{I}_i$, let $A(I_i)$ be the set of possible actions at $I_i$. Formally, the definition is given as follows.

\begin{definition}[Operational Search Plan]
Given the MTG at node $v\in\mathcal{V}$, an operational search plan for player $i \in \mathcal{N}$ is a function mapping each of his knowledge set to a probability distribution over the set of possible techniques at that knowledge set, given by:
\begin{equation}
b^v_i:\mathcal{I}_i \mapsto \bigcup_{I_i\in\mathcal{I}_i} \Delta(A(I_i)),
\end{equation}
such that $b^v_i(I_i)\in \Delta(A(I_i))$ for all $I_i\in\mathcal{I}_i$. Denote $B^v_i$ as the all admissible set of operational search policies of player $i\in\mathcal{N}$ at this MTG.
\end{definition}


In this work, we assume that all the players have \textbf{perfect recall}; i.e., the player remembers every piece of information that he knows from the past, including his moves, the other player’s moves, or chance moves. Under this assumption, we can always find the equivalence between the operational search plan is equivalent to the mixed penetration plans at the MTG of each node.

\begin{theorem}[Planning Equivalence]\label{theorem:equal}
    In every MTG in extensive form, if player $i\in\mathcal{N}$ has perfect recall, then for every mixed penetration plan there exists an equivalent operational search plan, and vice versa.
\end{theorem}

\begin{proof}
    For interested readers, the proof of the theorem follows Kuhn's theorem \cite{aumann1961mixed,kuhn1953extensive} in extensive-form games.
\end{proof}

Theorem~\ref{theorem:equal} indicates the equivalence between the mixed penetration plan and the operational search policy. The mixed penetration plan can be reduced to an operational search policy, and conversely, the operational search policy can generate a mixed penetration plan. The mixed penetration plan provides a holistic offline view, assigning a probability to each possible sequence of interactions. In contrast, the operational search policy describes the online decision-making process of the players. The equivalence allows us to choose the appropriate plan for the corresponding security purpose. The theorem, on the one hand, indicates that we can synthesize an operational strategy once we compute or are given a penetration plan. On the other hand, it suggests that a penetration plan or the course of actions of an attacker can be constructed after the penetration testing by obtaining the attacker’s strategy at each decision point.

\subsection{Security Schemes and Solution Concepts}\hfill

Depending on the security goals, our framework is able to describe different security schemes and provide corresponding solution concepts.

\subsubsection{Optimal Local Penetration Plan}\hfill

The primary goal of penetration testing is to identify vulnerabilities in the target system. This includes not only surface-level vulnerabilities that can be detected through vulnerability scanning but also deeper vulnerabilities that can only be discovered through a sequence of attack actions. Penetration testing provides a thorough examination of the system, and this type of security scheme is commonly known as red teaming. Red teaming involves simulating a malicious attacker to assess the effectiveness of the current defense policy. In this context, the defender's strategy remains fixed, while the attacker responds optimally to the defense policy. Red teaming aims to determine the optimal local penetration plan $\sigma^{v,red}_a$ that maximizes the attacker's utility, given the defender's strategy and the inherent system randomness in the system. The solution concept for the optimal local penetration plan is defined as follows:

\begin{definition}[Optimal Local Penetration Plan]\label{def:olap}
    For the MTG at node $v\in\mathcal{V}$, given the defense strategy $\sigma^v_d$ and system randomness $\sigma^v_c$, the optimal local penetration plan is a probability distribution over all attacker's pure penetration plans, i.e., $\sigma^v_a \in \Delta(Q^v_i)$ given by
\begin{align}
\sigma^{v,red}_a(\sigma^v_d,\sigma^v_c) \in \arg\max_{\sigma^v_a\in\Sigma^v_a}u^v_a(\sigma^v_a, \sigma^v_d,\sigma^v_c),
\end{align}
where $u^v_a(\sigma^v_a, \sigma^v_d,\sigma^v_c)$ is the expected utility of outcome generated following the plan profile $\Phi^v = (\sigma^v_a, \sigma^v_d,\sigma^v_c)$.
\end{definition}

The optimal local penetration plan is an ex-ante strategy where the attacker or penetration tester has complete information about the local node, including defense strategy, possible vulnerabilities, and system randomness. With this information, an optimal pure or mixed penetration plan can be obtained to visualize different attack chains along with their outcomes and probabilities. However, in practice, the attacker or penetration tester may not have complete information. Their penetration plan is developed through learning (e.g., via machine learning or reinforcement learning) without complete prior knowledge about the local node. The practically used penetration plan belongs to the family of operational search plans, which is a mapping from the knowledge set to the probability distribution over the set of possible techniques. 

\begin{remark}[Optimal v.s. Practical]
The practically used penetration plan is equivalent to the optimal local penetration plan when the penetration tester's learning results are perfect. This is possible when the penetration testing agent or attacker, through its learning process, has identified the exact set of actions to take in each state to maximize the expected reward, as if it had known the full model from the beginning. Under these conditions, the practically used penetration plan is identical to the optimal operational search plan. According to Theorem~\ref{theorem:equal}, this is thus equivalent to the optimal local attack policy in Def.~\ref{def:olap}. 
\end{remark}

The optimal local penetration plan generates useful byproducts that help describe the penetration plan. One such byproduct is the \textbf{\textit{course of action}}, which describes the realized sequence of attack techniques derived from the penetration plan.
Another important concept is the \textbf{\textit{tactic outcome probability}}, which represents the total probability of reaching any outcome of the game $z\in\mathcal{Z}^v$. Let $H^z\subset\mathcal{H}$ be the set of leaf vertices with outcome $z\in\mathcal{Z}^v$. Define $L(h^z)=\{(h_1,a_1),(h_2,a_2),\dots\}$ as the sequences of vertices and actions leading to the leaf vertex $h^z\in H^z$. The tactic outcome probability is defined as follows.

\begin{definition}[Tactic Outcome Probability]\label{def:outpr}
For the MTG at node $v\in\mathcal{V}$, given the nature's fixed policy (if any) and the plan profile of the attacker and the defender, i.e., $\Phi^v = (\sigma^v_a,\sigma^v_d,\sigma^v_c)$, we define $\tau^v: \mathcal{Z}^v \mapsto [0,1]$ as the tactic outcome probability. We use $\tau^v(z)$ to denote the probability of reaching outcome $z\in\mathcal{Z}^v$ as
\begin{align}
\tau^v(z\mid \Phi^v) =
= \sum_{h^z\in\mathcal{H}^z} \left[ \prod_{(h_j,a_j)\in L(h^z)} \sigma^v_i (a_j)\mathbbm{1}_{\{P(h_j)=i\}}\right],
\label{eq:tau}
\end{align}
where $P(h_j)$ is the turn function and $\sigma^v_i (a_j)$ is the probability that action $a_j$ is chosen by player $i=P(h_j)$.
\end{definition}

\subsubsection{Purple Teaming Defense Plan}\hfill

Purple teaming is a collaborative cybersecurity assessment that combines attack and defense strategies to enhance the overall security posture of a system. While red teaming penetration testing predicts the attacker's behavior, purple teaming focuses on improving the defense policy to mitigate potential attacks. This approach corresponds to a Stackelberg game or leader-follower model, where the defender, as the leader, enforces their strategy on the attacker, as the follower. The defender must anticipate the attacker's responses to the defense strategy and optimize the defense policy accordingly, resulting in a bi-level optimization problem. Penetration testing provides credible predictions of the attacker's penetration plan, enabling proactive defense with purple teaming. The solution concept for the purple teaming defense plan is defined as follows:

\begin{definition}[Optimal Purple Teaming Defense Plan]
    For the MTG at node $v\in\mathcal{V}$, given the system randomness $\sigma_c^v\in$, the optimal purple teaming defense plan includes two parts: $\sigma^{v,pur}_d\in\Sigma_d^v$ is the optimal purple teaming defense plan, which is a probability distribution over all defender's pure defense plans, i.e.,$\sigma_d^{v,pur}\in\Delta(Q^v_d)$; $\sigma^{v,*}_a\in\Sigma_a^v$ is the anticipated optimal local penetration plan for the attacker given the defense plan. 
    \begin{align}
    \sigma^{v,pur}_d(\sigma^v_c) \in & \max_{\sigma^v_d \in \Sigma^v_d}\quad  u^v_d(\sigma^{v,*}_a, \sigma^v_d,\sigma^v_c) \\
    \text{s.t.}\quad & \sigma^{v,*}_a \in \arg\max_{\sigma^v_a\in\Sigma^v_a}u^v_a(\sigma^v_a, \sigma^v_d,\sigma^v_c).
    \label{eq:br}
\end{align}
\end{definition}

The inner optimization problem aligns with the optimal local penetration plan as defined in Def.~\ref{def:olap}, aiming to predict the worst-case attacker behavior under the current defense strategy. Penetration testing, utilizing learning techniques, determines the attacker's anticipated response to a given defense strategy. To implement purple teaming defense in practical settings, organizations undergo an iterative process where the defender tests a defense strategy, observes the worst-case attack, and then adjusts the defense to achieve better utility.

\subsubsection{Risk Assessment at Equilibrium}\hfill

Another important venue penetration testing contributes to is the risk assessment of the system. Instead of focusing on individual attack events, a risk assessment would take into account the average or the steady state of the long-term behaviors of the attacker and defender in the long run. The concept of equilibrium in game theory offers a natural way to analyze these steady-state strategic interactions within the system. A Nash Equilibrium (NE) in the MTG provides a solution where no player has an incentive to deviate from their strategy. Formally, the solution concept for risk assessment is defined as follows:

\begin{definition}[Nash Equilibrium-Informed Risk Assessment]
For the micro tactic game at node $v\in\mathcal{V}$, given the system randomness $\sigma_c^v\in\Sigma^v_c$, the Nash equilibrium-informed risk assessment is a plan profile $(\sigma_a^{v,*}, \sigma_d^{v,*})$, where $\sigma_a^{v,*}$ is the equilibrium penetration plan for the attacker and $\sigma_d^{v,*}$ is the equilibrium defense plan for the defender. The Nash equilibrium plans satisfy
\begin{align}\label{NE}
    u^v_i(\sigma_i^{v,*}, \sigma_{-i}^{v,*})\geq u^v_i(\sigma_i, \sigma_{-i}^{v,*}),
\end{align} 
for all admissible $\sigma^v_i\in \Sigma^v_i$ and for all $i\in \mathcal{N}$.
\end{definition} 

To practically solve the game, we consider a refinement of Nash equilibrium in sequential games: Subgame Perfect Nash Equilibrium (SPNE). In addition to satisfying the conditions of Nash equilibrium, SPNE requires that strategies remain in equilibrium at every possible subgame of the overall game. It can be solved using backward induction as the game-theory version of the dynamic programming principles.


\begin{theorem}\label{theorem:spne}
    For every finite micro tactic game at node $v\in\mathcal{V}$ with fixed system randomness $\sigma_v^c\in\Sigma_c^v$, the game with perfect recall has a subgame perfect Nash equilibrium in mixed or operational search penetration/defense plans. The game with perfect information has a subgame perfect Nash equilibrium in pure penetration and defense plans \cite{maschler2020game}.
\end{theorem}

Theorem~\ref{theorem:spne} states that we can always find the risk assessment equilibrium in mixed plans or operational search policies, even with imperfect information. Mixed penetration plans provide the probability of the entire attack/defense action sequence occurring, offering a holistic view for analysis purposes. On the other hand, operational search policies focus on what happens in each knowledge set, providing a fine-grained strategy. The equivalence between them allows us to zoom in or out as needed, facilitating flexible and comprehensive analysis.


\subsection{Macro Strategic Process}

One key component in the MTG is the utility function for each outcome, $u_i^v(z)$, for all $z \in \mathcal{Z}^v$ and $i\in\mathcal{N}$. Utilities represent the payoff or cost of staying or moving to the next node and must be evaluated globally, considering neighboring nodes and their connections. After local exploration and exploitation, the attacker can use obtained credentials or discovered vulnerabilities to move to different nodes, a process known as \textit{lateral movement}. The attacker's movement and the creation of the attack kill chain depend on the network topology and the expected utilities of each node. We model this decision-making process across the network using an MDP, referred to as a Macro Strategic Process (MSP). 

\subsubsection{MSP Modeling}\hfill
\begin{definition}[Macro Strategic Process (MSP)]
\label{def:msp}
The Macro Strategic Process of the MEGA-PT is defined by an MDP $\Lambda^g$. Given the target networked system $\mathcal{G}=\langle \mathcal{V}, \mathcal{E}\rangle$, the MSP for the attacker can be represented by a tuple $\Lambda^g = \langle  \mathcal{S},\mathcal{A}^g, T, R, \gamma \rangle$, where each component represents:
\begin{itemize}
    \item \textbf{Network Nodes $\mathcal{S} = \mathcal{V}$:}  The nodes in the network form the state space. Each node or state $v\in\mathcal{V}$ processes an MTG $\Gamma^v$ as defined in Def.~\ref{def:mtg}.
    \item \textbf{Connections $\mathcal{A}^g = \mathcal{E}$:} The connections between the nodes are the attacker's action space, which is equivalent to all the directed edges in the network.
    \item \textbf{Transition Success Probability $T:\mathcal{S}\times\mathcal{A}^g \mapsto \Delta(\mathcal{S})$:} This function describes the success rate of the lateral movement attempt between the nodes. 
    \item \textbf{Movement Rewards $R:\mathcal{S}\times \mathcal{A}^g\times \mathcal{S} \mapsto \mathbb{R}$:} The immediate reward or penalty for the attacker when trying to laterally move along the edge to the other node. 
    \item \textbf{Discounting Factor $\gamma\in (0,1]$.}
\end{itemize}
\end{definition}

From a global perspective, each node in the network $v\in\mathcal{V}$ can be viewed as a state in the Markov Decision Process. The edges in the network indicate the lateral movement of the attacker within the network. Whether the attack attempt is successful depends on the capability of the attacker. For simplicity, in this work, we assume the transition success probability is defined as follows. For every $s\in\mathcal{V}$ and $a^g\in\mathcal{A}^g$,
\vskip -3mm
\begin{equation}
T(s^\prime \mid s = v, a^g = (v,u)) =
\begin{cases}
1 & \text{if $u = v$ and $s^\prime = v$,} \\
c_a & \text{if $u \neq v$ and $s^\prime = u$,}\\
1-c_a & \text{if $u \neq v$, and $s^\prime = v$,}\\
0 & \text{otherwise.}
\end{cases}
\end{equation}

If the attacker chooses to stay at the same node, the self-loop edge will lead to the same state with  probability one. If the attacker chooses to use any outgoing edge and move to another node, the attempt will succeed with probability $c_a \in [0,1]$, which represents the attacker's capability. If the attempt fails, the attacker will stay at the same node.

The goal of penetration testing is to estimate the potential damage an attacker can inflict by compromising the network and affecting system production. A positive reward is given when the attacker enters a node, with the reward value depending on the node's importance to the system. Conversely, staying at the same node indicates that the attacker either failed to move to another node or that the information obtained from the MTG was insufficient for progression. Therefore, staying at the same node results in a negative penalty for the attacker. The movement reward function is given by the following equation:
\begin{equation}
R(s=v,a^g=(v,u),s^\prime) = \begin{cases}
M_a & \text{when } s'  = v,\\
\Bar{V}(v) & \text{when } s' = u \,,\, \forall  u\in\mathcal{V}\setminus \{v\}.
\end{cases}
\end{equation}
where $M_a\in\mathds{R}^-$ is a penalty for the attacker staying at the same node without progressing towards the target. $\Bar{V}: \mathcal{V}\mapsto \mathds{R}^+$ is the reward for entering the state. This value depends on the production importance of the node $v\in\mathcal{V}$ to the target system.

\subsubsection{Global Attack Strategy}\hfill

Unlike traditional MDPs, where the attacker can freely choose actions to optimize expected utility, in the realistic penetration testing settings, the attack strategy at the network level depends on explorations at the local nodes. If the attacker does not find any vulnerabilities leading to the next node, they cannot move forward. Therefore, the global attack strategy in the MSP relies on the outcomes of the MTG.

For each MTG $\Gamma^v$ at node $v\in\mathcal{V}$, the optimal local penetration plans generate the tactic outcome probability as defined in Def.~\ref{def:outpr}. 
Since the outcome space of the MTG is equivalent to the set of the outgoing edges at node $v$, i.e., $ \mathcal{Z}^v= \{u\mid u\in\mathcal{V},(v,u)\in\mathcal{E}\}$, we can view the tactic outcome probability $\tau^v(z)$ as the probability that the attacker will choose action $a^g = (v,z)$ for the MSP. Formally, it leads to the following definition.

\begin{definition}[Global Attack Strategy]
\label{def:global}
Consider the MTG defined in Definition~\ref{def:mtg} and the MSP defined in Definition~\ref{def:msp}. The global attack strategy in MSP is a mapping from the state space to the global action space, i.e., $\pi^g: \mathcal{S}\mapsto \Delta(\mathcal{A}^g)$. For node $v\in\mathcal{V}$, given the MTG $\Gamma^v$ and the local plan profile $\Phi^v = ( \sigma^v_c,\sigma^v_a,\sigma^v_d)$, the global attack strategy is given by
\begin{align}
    \pi^g(a^g\mid s) = \pi^g(a^g = (v,z)\mid s = v) 
    = \tau^v(z\mid \Phi^v),\qquad \forall z\in\mathcal{Z}_v,
    \label{eq:policy}
\end{align}
where $\tau^v(z\mid \Phi^v)$ is the tactic outcome probability as defined in Def.~\ref{def:outpr}.
\end{definition}

The global attack strategy in the MTG outlines the cyber kill chain and the sequence of tactics across the entire system. Rather than focusing on the details at each local node, the MSP connects all the nodes in the network, offering a comprehensive risk assessment for the entire system. This holistic approach allows organizations to better understand the interconnections within their network and the cascading effects of vulnerabilities throughout the system.

\subsection{Meta Penetration Game and Playbook}

Policy evaluation offers a way to estimate the effectiveness of the global attack strategy $\pi^g$ in terms of expected cumulative utilities. Similar to traditional MDPs, policy evaluation of the global attack strategy computes the value functions using the Bellman equations. For all states $s\in\mathcal{S}$, the value function under $\pi^g$ is given by
\begin{equation}
    V^{\pi^g}(s) = \sum_{a^g\in\mathcal{A}^g} \pi^g(a^g\mid s)\sum_{s'\in\mathcal{V}}T(s'\mid s,a^g)\left[R(s,a^g,s')+\gamma V^{\pi^g}(s')\right].
    \label{eq:value}
\end{equation}

The value at each node in the system describes the expected return starting from that node and then acting according to the global attack strategy $\pi^g$. 
For the players at the MTG, the utility of each outcome describes the expected reward of taking that action and moving to the next node in the macro strategy process. Thus, we define the utility functions in the MTG as follows.

\begin{definition}[MTG Utilities]
\label{def:utility}
Given the global attack strategy $\pi^g\in\Pi^g$, the attacker's utility functions of reaching outcome $z\in\mathcal{Z}^v$ in the MTG at node $v\in\mathcal{V}$ are defined as the 
\begin{align}
    u^v_a(z=u)  
    = \sum_{s'\in\mathcal{V}}T(s'\mid s=v,a^g=(v,u))\left[R(s,a^g,s')+\gamma V^{\pi^g}(s')\right],
\end{align}
where $V^{\pi^g}$ is the policy evaluation value function in \eqref{eq:value}. The defender's utility is the opposite of the attacker, i.e., $u^v_d(z) = -u^v_a(z)$  for all $z\in\mathcal{Z}^v$.
\end{definition}

    \begin{figure}[!t]
    \centering
    \includegraphics[width=0.9\linewidth]{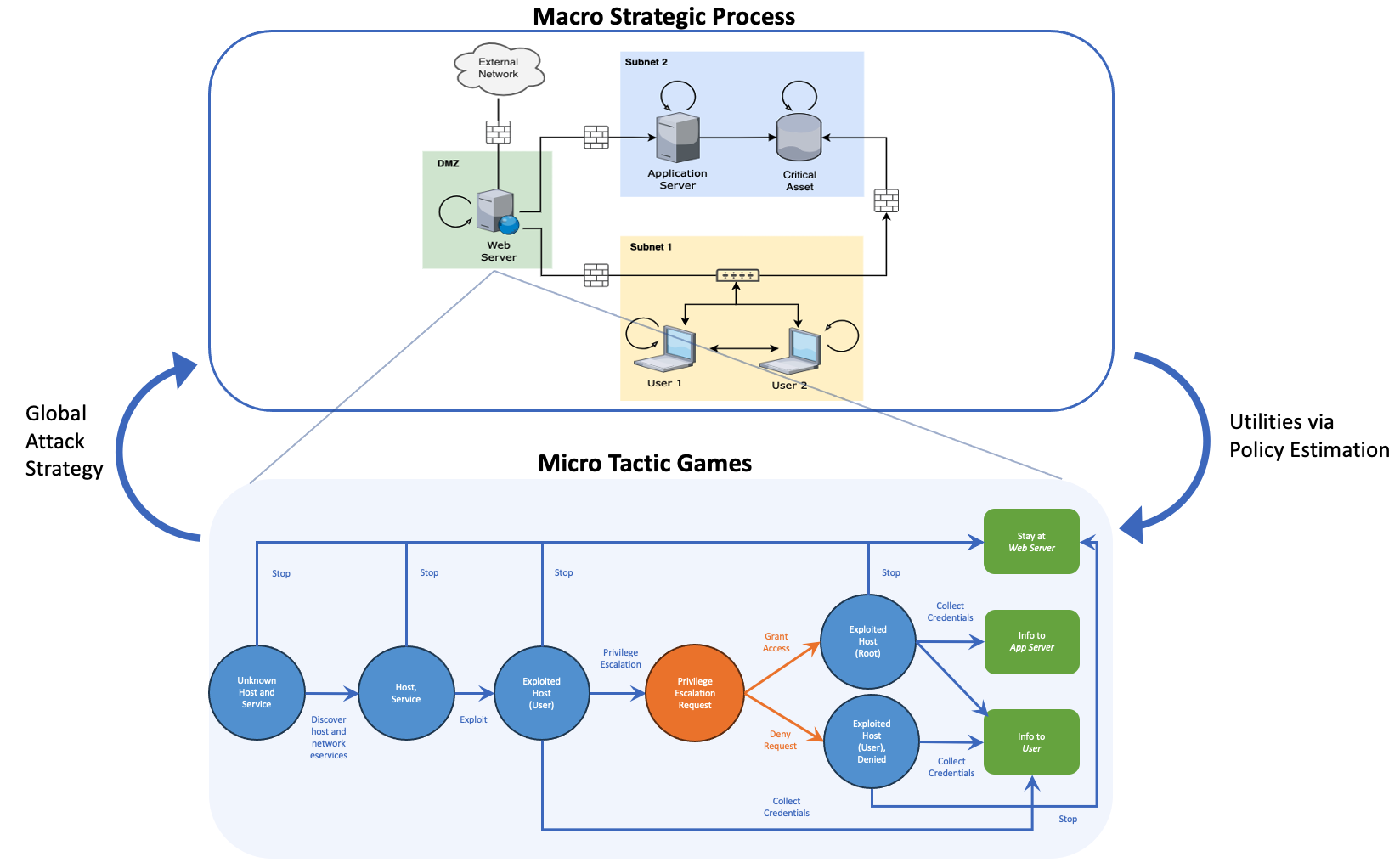}
    \caption{Relationship between Macro Strategic Process and Micro Tactic Games. The local penetration plans in the micro games  affect the global attack strategy, while the policy evaluation at the macro process helps provide the utilities in the micro games.}
    \label{fig:relationship}\vskip -5mm
    \end{figure}

Figure~\ref{fig:relationship} illustrates the relationship between the MSP and the MTGs. The MSP defines the global attack strategy, forming an attack kill chain and providing estimated values for each node through policy evaluation under the current strategy. These estimated values represent the expected outcome utilities at each MTG, guiding the formulation of detailed penetration plans at each node. The sequence of attack and defense techniques at the local node influences the global attack strategy in the macro view, emphasizing how lateral movement is determined by exploration and exploitation. This iterative process continues until a meta-solution is reached. Together, the MSP and the MTGs constitute a meta-security game that captures decision-making in penetration testing at both network and node levels.

\begin{definition}[Meta-Security Game]\label{def:metasec}
Given the network system graph $\mathcal{G} = \langle \mathcal{V}, \mathcal{E} \rangle$, the meta-security game is composed of two parts: $\Xi = \langle \{\Gamma^v\}_{v\in\mathcal{V}}, \Lambda^g \rangle$, where $\{\Gamma^v\}_{v\in\mathcal{V}}$ is the set of MTGs as defined in Definition~\ref{def:mtg} and $\Lambda^g$ is the macro strategy process as defined in Definition~\ref{def:msp}. 
\end{definition}

The MSP and the MTGs are inherently coupled, as the local penetration plans in the MTGS naturally affect the global attack strategy, while the policy evaluation at the MSP helps provide the utilities in the MTGs. Hence, a holistic solution concept is necessary for the proposed meta-security game.

\begin{definition}[Meta Penetration Playbook]\label{def:playbook}
    Consider the meta-security game $\Xi = \langle \{\Gamma^v\}_{v\in\mathcal{V}}, \Lambda^g \rangle$ defined in Definition~\ref{def:metasec}, the meta penetration playbook $\xi=\langle  \{\Phi^v\}_{v\in\mathcal{V}}, \pi^g \rangle$ is composed of two elements:
    \begin{enumerate}[nosep, leftmargin=*]
        \item[\textbullet] \textbf{Local Penetration profile:} $\Phi^v = (\sigma^{v,*}_a,\sigma^{v,*}_d,\sigma^v_c)$ constitutes the local penetration plans of all players for the MTG at node $\Gamma^v$ for each $v\in\mathcal{V}$,
        \item[\textbullet] \textbf{Global Attack Strategy:} $\pi^g$ is the global attack strategy in the macro strategy process,
    \end{enumerate}
    which satisfy two conditions:
    \begin{enumerate}[nosep, leftmargin=*]
        \item[\textbullet]  \textbf{Policy Dependency:} 
        The global attack strategy $\pi^a$ at the macro strategy process depends on the local penetration plans $\{\Phi^v\}_{v\in\mathcal{V}}$ as defined in 
        Definition~\ref{def:global},
        \item[\textbullet] \textbf{Value Dependency:}
        For each MTG at node $v\in\mathcal{V}$, the utility of each tactic's expected outcome depends on the policy evaluation results of global attack strategy $\pi^g$ according to Definition~\ref{def:utility}.
    \end{enumerate}
\end{definition}

In a global view, a complete cyber attack kill chain comprises a sequence of tactics. The global attack strategy guides how to compose this attack kill chain within the target system. Within each tactic, there is a sequence of techniques. The local penetration profile at each node describes the decision-making process of the players to complete these technique sequences. The policy and value dependencies connect the macro and micro solutions, helping us to form an efficient and consistent meta-penetration playbook for the meta-security game.

\section{Computation}

To determine the optimal meta-penetration playbook, we propose the following algorithm to find the exact solution. In this section, we use the purple teaming defense as the penetration scheme and solution concept to describe the computational process. For other security schemes, the general structure of the algorithm remains the same, but the method for obtaining the local penetration profile in each MTG differs (line $7$ in Algorithm~\ref{alg:compute}).

{\renewcommand{\arraystretch}{1.5}
\setlength{\textfloatsep}{2pt}
\begin{algorithm}[!t]
\caption{Purple Teaming Meta Penetration Playbook Algorithm}\label{alg:compute}
\KwIn{Meta-security game $\Xi=\langle \{\Gamma^v\}_{v\in\mathcal{V}},\Lambda^g\rangle$}
\vskip 1mm
Set the utilities $u^v_i$ in each MTG to arbitrary value\;
\Repeat{Meta penetration playbook converges}{
\textbf{Micro Penetration Profile Computation: }\\ 
\hspace{2em} For every MTG at $v\in\mathcal{V}$, compute the purple teaming penetration plan profile $\Phi^v =(\sigma^{v,*}_a,\sigma^{v,pur}_d,\sigma^v_c)$\;
\hspace{2em} Compute the attack strategy $\pi^a$ under $\{\Phi^v\}_{v\in\mathcal{V}}$ using \eqref{eq:policy} \;
\textbf{Macro Attack Strategy Evaluation:}\\ 
\hspace{2em} Compute the value function $V^{\pi^g}$ of $\Lambda^g$ using \eqref{eq:value}\;
\hspace{2em} Update the utilities $u^v_i$ in each MTG $\Gamma^v$\;}
\vskip 1mm
\KwResult{Meta penetration playbook $\xi=\langle  \{\Phi^v\}_{v\in\mathcal{V}}, \pi^g \rangle$.} 
\end{algorithm}} 

To analyze the risks of each node and evaluate the effectiveness of the system defense, we define the network risk score as a measurement metric. For each node $v\in\mathcal{V}$ in the system, we are interested in whether the attacker has access to this node, and what is the expected damage he can create. Let $V_{max}\in\mathbb{R}^+$ be the maximum damage that the attacker could cause. Given the meta-security game $\Xi$ and the corresponding meta penetration playbook $\xi$, the network risk score of node $v\in\mathcal{V}$ is a normalized risk value $Risk(v\mid \xi)\in[0,1]$ given by $V^{\pi^g}(v)/V_{max}$ if $V^{\pi^g}(v)$ is non-negative; otherwise the score is set to $0$.


\section{Case Study}

We use the network topology depicted in Figure~\ref{fig:example} as a case study to demonstrate the effectiveness of MEGA-PT. The system consists of $5$ nodes, including the web server, two user devices, the application server, and the critical asset. The MTG trees for each node are illustrated in Appendix~\ref{sec:app}. These game trees align with attack scenarios from the MITRE ATT\&CK model and can be adjusted to fit specific system structures. 
We evaluate the performance of our model through numerical experiments conducted in a self-built Python simulator. While the model's applicability extends to practical systems given the network topology and vulnerability trees, the details on how to gather this information are beyond the scope of this paper.

The penetration testing agent acts as an attacker entering the system from the external network, starting at the web server. The goal is to penetrate the system and potentially affect operations at the critical asset. We assume that there is an artificial node in the network representing a successful compromise of operations. Once the attacker reaches this node, the penetration process is considered terminated.
In our experiments, we set the parameters as follows: the immediate rewards for entering each node and the penalty are specified in Table~\ref{table:values}. The attacker's capability is denoted as \( c_a = 0.8 \) by default, and we use \( \gamma = 0.9 \) for the policy evaluation process.

\begin{table}[!t]
    \caption{Movement rewards for the attacker in the network.}
    \centering
    \begin{tabular}{|c|c|c|c|c|c|}
        \hline
        Web Server   & User Devices& App Server & Critical Asset & Operation Down & Penalty\\
        \hline
        0 & 5 & 20 & 30 & 100 & -15\\
        \hline
    \end{tabular}
    \vskip -5mm
    \label{table:values}
\end{table}

\subsection{Optimal Penetration Plan and Purple Teaming}

\begin{figure}[t!]
\centering     
\subfloat[Fixed defense, Weak Attacker $c_a = 0.2$.]{\includegraphics[width=0.5\textwidth]{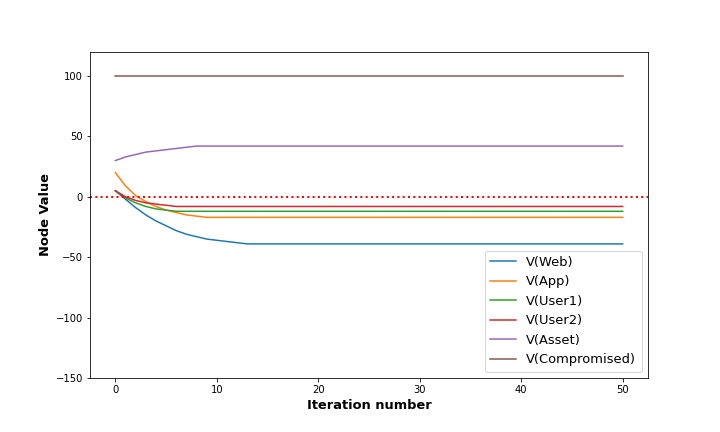}}
\subfloat[Purple Teaming. Weak Attacker $c_a = 0.2$.]{\includegraphics[width=0.5\textwidth]{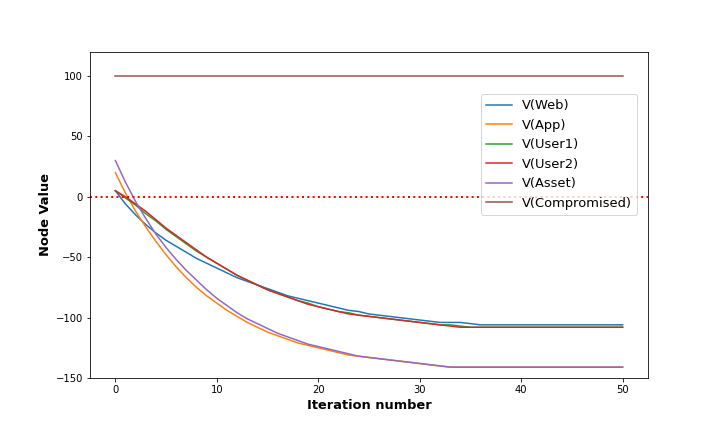}}\\
\subfloat[Fixed defense, Median Attacker $c_a = 0.5$.]{\includegraphics[width=0.5\textwidth]{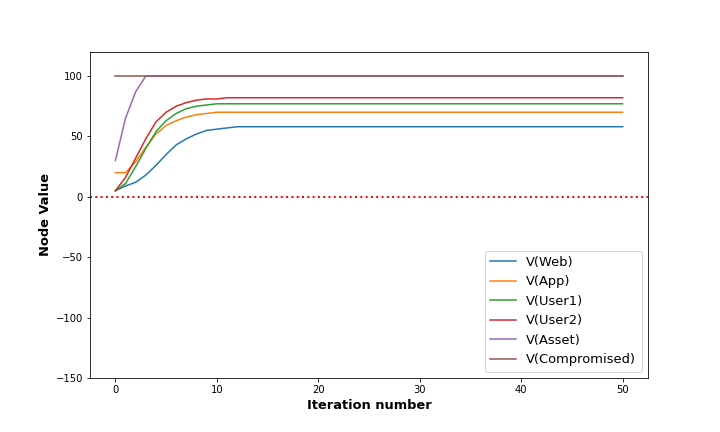}}
\subfloat[Purple Teaming. Median Attacker $c_a = 0.5$.]{\includegraphics[width=0.5\textwidth]{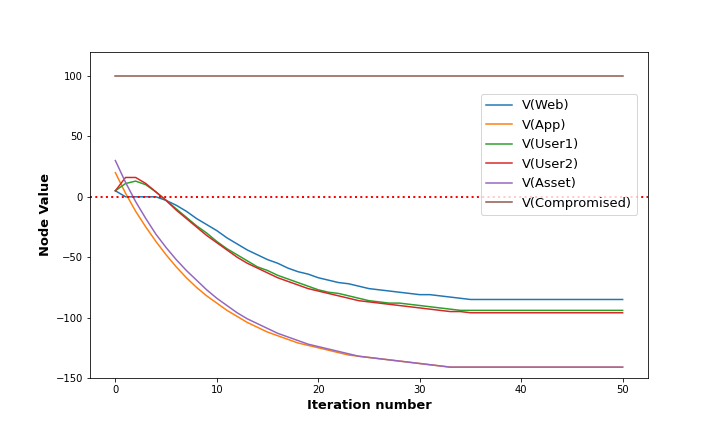}}\\
\subfloat[Fixed defense, Strong Attacker $c_a = 0.8$.]{\includegraphics[width=0.5\textwidth]{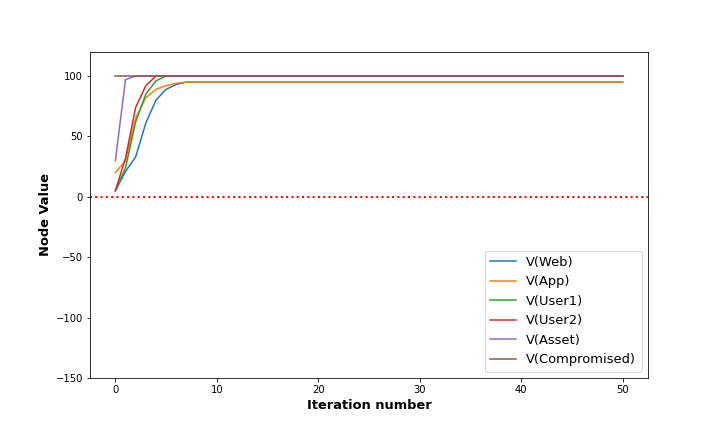}}
\subfloat[Purple Teaming. Strong Attacker $c_a = 0.8$.]{\includegraphics[width=0.5\textwidth]{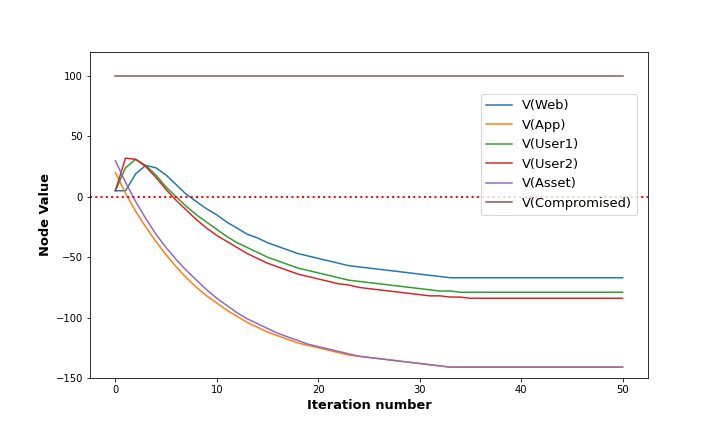}}
\caption{Node values under different conditions. The $x$-axis is the number of iterations and $y$-axis is the value of the node for the attacker. We consider both fixed defense and purple teaming defense against three types of attackers ($c_a=\{0.2,0.5,0.8\}$). }
\vskip -10mm
\label{fig:opt}
\end{figure}

Figure~\ref{fig:opt} illustrates the value of each node in the network during the meta penetration playbook computation. We consider three types of attackers with different capabilities: a weak attacker with \( c_a = 0.2 \), a median attacker with \( c_a = 0.5 \), and a strong attacker with \( c_a = 0.8 \). In the left column, we test the model under a fixed defense strategy. Specifically, at the web server, the probability of the defender granting access is \( 0.7 \). At the application server, the probability of the defender enforcing strict authorization policies is \( 0.3 \). Finally, at the critical asset, the probability of the defender executing the command is \( 0.6 \).

The value of each node represents the expected accumulated reward if the attacker starts penetration from that node. As observed in the figures, when the attacker is weak, even with a fixed defense strategy, no node in the system yields a positive reward. However, as the attacker's capabilities increase, certain states in the system can provide positive rewards. The stronger the attacker, the higher the maximum reward achievable. 
The potential system damage can be mitigated by adopting a purple teaming defense strategy. In Figure~\ref{fig:opt}, the right column demonstrates that when implementing purple teaming and adjusting the defense plan, scenarios where the attacker previously gained positive rewards turn into negative rewards as the strategy converges. The attacker can only achieve gains when the defense strategy is not yet converged. These results illustrate that our model effectively accommodates varying attacker capabilities, and purple teaming offers enhanced defense capabilities.

    \begin{figure}[!t]
    \centering
    \includegraphics[width=0.7\linewidth]{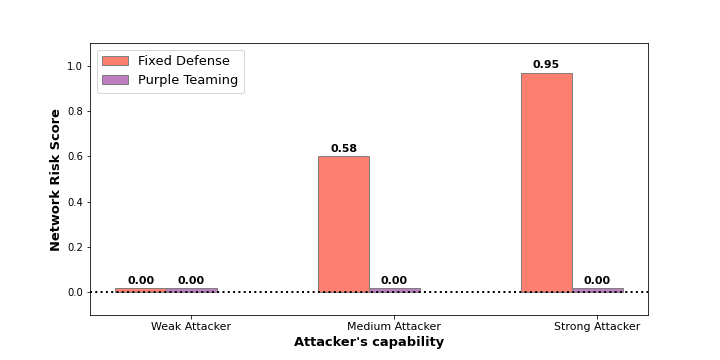}
    \caption{Network risk score of the web server under different security schemes.} \vskip -5mm
    \label{fig:score}
    \end{figure}

Figure~\ref{fig:score} shows the network risk score of the web server under different security schemes. We consider \( V_{max} = 100 \) as the maximum damage the attacker could cause. With a fixed defense, the web server's network risk score increases as the attacker's capability increases. The web server is not risky only when the attacker is weak, and his capability is low (\( c_a = 0.2 \)). With purple teaming defense, the web server is safe against all types of attackers, and the network risk score always remains zero. This indicates that purple teaming defense helps the system find a defense strategy that can reduce the network risk and prevent the system from being compromised.

\subsection{Vulnerability Adaptability}\hfill

We consider a scenario where the MTG changes within the network. In this experiment, we assume a change at the application server where no longer any information can lead to the critical asset. This change might occur if the system encrypts this information or restores the database, preventing the attacker from decrypting or accessing the data related to the critical asset. Consequently, the MTG at the application server changes, with the only outcome being staying at the same node after exploration.

We compute the meta penetration playbook under a fixed defense strategy before and after the local node change:
\begin{align*}
    \text{Before: } \xi &= \langle \Phi^{web},  \textcolor{blue}{\Phi^{app}}, \Phi^{user1}, \Phi^{user2}, \Phi^{asset}, \textcolor{blue}{\pi^g} \rangle;\\
    \text{After: } \xi^\prime &= \langle \Phi^{web},  \textcolor{red}{\Phi^{app,\prime}}, \Phi^{user1}, \Phi^{user2}, \Phi^{asset}, \textcolor{red}{\pi^{g,\prime}}\rangle
\end{align*}
It can be noticed that in the meta penetration playbook, the vulnerability change only affects the application server and the global attack strategy, while the other plans remain the same. This indicates that we only need to recompute the MTG at that node while retaining the original structure of the other nodes and updating the global attack policy. 

The comparison of global attack strategies before and after this local node change is illustrated in Table~\ref{tab:before} and Table~\ref{tab:after}. Prior to the vulnerability change, the attacker from the web server had a high probability ($\Pr=0.7$) of transferring to the application server to further attack the system. However, after the change, since there is no longer a connection between the application server and the critical asset, our computed penetration plan adjusts its global attack strategy and no node would transfer to the application server anymore. The attacker would focus solely on the user devices to find any information that could compromise the operation. This result demonstrates that MEGA-PT can effectively adapt to local vulnerability changes.

\begin{table}[!t]
    \centering
    \begin{minipage}[t]{0.45\textwidth}
        \centering
        \caption{Global attack strategy $\pi^g$ before vulnerability change.}
        \[
        \begin{array}{c|cccccc}
        \text{} & \text{web} & \text{app} & \text{user} & \text{user} & \text{asset} & \text{final} \\
        \hline
        \text{web} & 0 & 0.7 & 0.3 & 0 & 0 & 0 \\
        \text{app} & 0 & 0.37 & 0 & 0 & 0.63 & 0 \\
        \text{user} & 0 & 0 & 0 & 0.5 & 0.5 & 0 \\
        \text{user} & 0 & 0 & 0.3 & 0 & 0.7 & 0 \\
        \text{asset} & 0 & 0 & 0 & 0 & 0.4 & 0.6 \\
        \text{final} & 0 & 0 & 0 & 0 & 0 & 1 \\
        \end{array}
        \]
        \label{tab:before}
    \end{minipage}%
    \hfill
    \begin{minipage}[t]{0.45\textwidth}
        \centering
        \caption{Global attack strategy $\pi^g$ after vulnerability change.}
        \[
        \begin{array}{c|cccccc}
        \text{} & \text{web} & \text{app} & \text{user1} & \text{user2} & \text{asset} & \text{final} \\
        \hline
        \text{web} & 0 & 0 & 1 & 0 & 0 & 0 \\
        \text{app} & 0 & 1 & 0 & 0 & 0 & 0 \\
        \text{user1} & 0 & 0 & 0.5 & 0.5 & 0 & 0 \\
        \text{user2} & 0 & 0 & 0.3 & 0 & 0.7 & 0 \\
        \text{asset} & 0 & 0 & 0 & 0 & 0.4 & 0.6 \\
        \text{final} & 0 & 0 & 0 & 0 & 0 & 1 \\
        \end{array}
        \]

        \label{tab:after}
    \end{minipage}\vskip -8mm
\end{table}

\subsection{Network-level Scalability}\hfill

In this scenario, we demonstrate the scalability of MEGA-PT by increasing the number of user devices within the subnet. We assume that all users share the same micro tactic tree, and the outcome leads to other users randomly transferring to another user device in the network with equal probability. Since MEGA-PT is modular, it allows us to compute each micro tactic tree in parallel. This means that if users share the same micro game, we can compute one instance and apply the result to all nodes in the network without recomputation. In contrast, traditional reinforcement learning-based methods treat each node's status in the system as a separate state. As the number of users increases, the state space grows exponentially, resulting in significantly increased computational time. 

From a meta penetration playbook point of view, the network-level scale change results in the following change in the playbook:
\begin{align*}
    \text{Before: } \xi &= \langle \Phi^{web},  \Phi^{app}, \textcolor{blue}{\Phi^{user}}, \Phi^{asset}, \textcolor{blue}{\pi^g} \rangle;\\
    \text{After: } \xi^\prime &= \langle \Phi^{web},  \Phi^{app}, \textcolor{blue}{\Phi^{user}, \Phi^{user}, \Phi^{user},\dots}, \Phi^{asset}, \textcolor{red}{\pi^{g,\prime}}\rangle
\end{align*}
Since the user devices are of the same type, each local penetration profile for the user device is the same. The only updated element in the playbook is the global attack strategy as it would consider more nodes in the system.

Figure~\ref{fig:scale} compares the computational time for finding the optimal strategy between our model and an RL-based model under different numbers of users. In the RL-based model, the state is the aggregation of all related information at each node, such as whether the web server has been discovered or whether the user credential has been found. The transition and reward in the RL-based model follow the same setting, but the state and transition space are enormous. We use Q-learning as the learning method under fixed defense and compare the computational time to find the optimal penetration strategy in the system. It is evident that as the number of users increases, the computational time for the RL-based method increases drastically, whereas our method shows minimal change. These results demonstrate that MEGA-PT scales effectively with large networks containing similar devices, providing robust scalability.
    \begin{figure}[!t]
    \centering
    \includegraphics[width=0.7\linewidth]{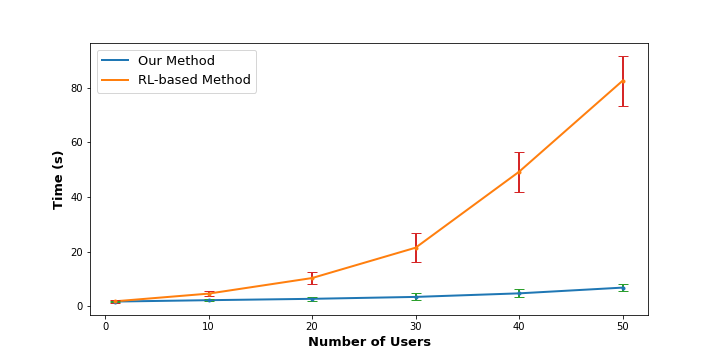}
    \caption{Scalability comparison between our model and RL-based model.}
    \label{fig:scale}\vskip -5mm
    \end{figure}

\section{Conclusion}
In this work, we propose MEGA-PT, a meta-game agile penetration testing model for automated and effective penetration testing. This model features MTGs for local node interactions and a macro strategy process for network-wide attack chains. It adheres to the TTPs in real cyber security frameworks, allows distributed and modularized penetration testing, and adapts to changes at both local and network levels. Experiments show that using MEGA-PT's purple teaming, the system can find effective defense strategies to reduce the network risk score of each node. Compared to other RL-based automated penetration testing models, MEGA-PT's distributed features enable agile adaptation to both local-level vulnerability changes and network-level topology changes, allowing effective and scalable penetration testing in large network systems.

For future work, we plan to discuss global defense strategies at the macro level and explore partial information in the game. MEGA-PT is promising for extension to different security schemes and can serve as a foundational framework for the next generation of automated penetration testing.

\bibliographystyle{splncs04}
\bibliography{ref.bib}

\begin{thebibliography}{1}
\providecommand{\url}[1]{\texttt{#1}}
\providecommand{\urlprefix}{URL }
\providecommand{\doi}[1]{https://doi.org/#1}

\bibitem{aumann1961mixed}
Aumann, R.J.: Mixed and behavior strategies in infinite extensive games. Princeton University Princeton (1961)

\bibitem{ge2023gazeta}
Ge, Y., Zhu, Q.: Gazeta: Game-theoretic zero-trust authentication for defense against lateral movement in 5g iot networks. IEEE Transactions on Information Forensics and Security  (2023)

\bibitem{ghanem2019reinforcement}
Ghanem, M.C., Chen, T.M.: Reinforcement learning for efficient network penetration testing. Information  \textbf{11}(1), ~6 (2019)

\bibitem{hu2020automated}
Hu, Z., Beuran, R., Tan, Y.: Automated penetration testing using deep reinforcement learning. In: 2020 IEEE European Symposium on Security and Privacy Workshops (EuroS\&PW). pp. 2--10. IEEE (2020)

\bibitem{kuhn1953extensive}
Kuhn, H.W.: Extensive games and the problem of information. Contributions to the Theory of Games  \textbf{2}(28),  193--216 (1953)

\bibitem{maschler2020game}
Maschler, M., Zamir, S., Solan, E.: Game theory. Cambridge University Press (2020)

\bibitem{mitre2020mitigations}
{MITRE}: Mitigations enterprise mitre att\&ck. \url{https://attack.mitre.org/mitigations/enterprise/} (2020)

\bibitem{shmaryahu2017partially}
Shmaryahu, D., Shani, G., Hoffmann, J., Steinmetz, M.: Partially observable contingent planning for penetration testing. In: Iwaise: First international workshop on artificial intelligence in security. vol.~33 (2017)

\bibitem{zhao2021combating}
Zhao, Y., Ge, Y., Zhu, Q.: Combating ransomware in internet of things: a games-in-games approach for cross-layer cyber defense and security investment. In: International Conference on Decision and Game Theory for Security. pp. 208--228. Springer (2021)

\end{thebibliography}

\appendix
\section{Appendix: Micro Tactic Game Trees}
\label{sec:app}


    \begin{figure}[!h]
    \centering
    \includegraphics[width=0.9\linewidth]{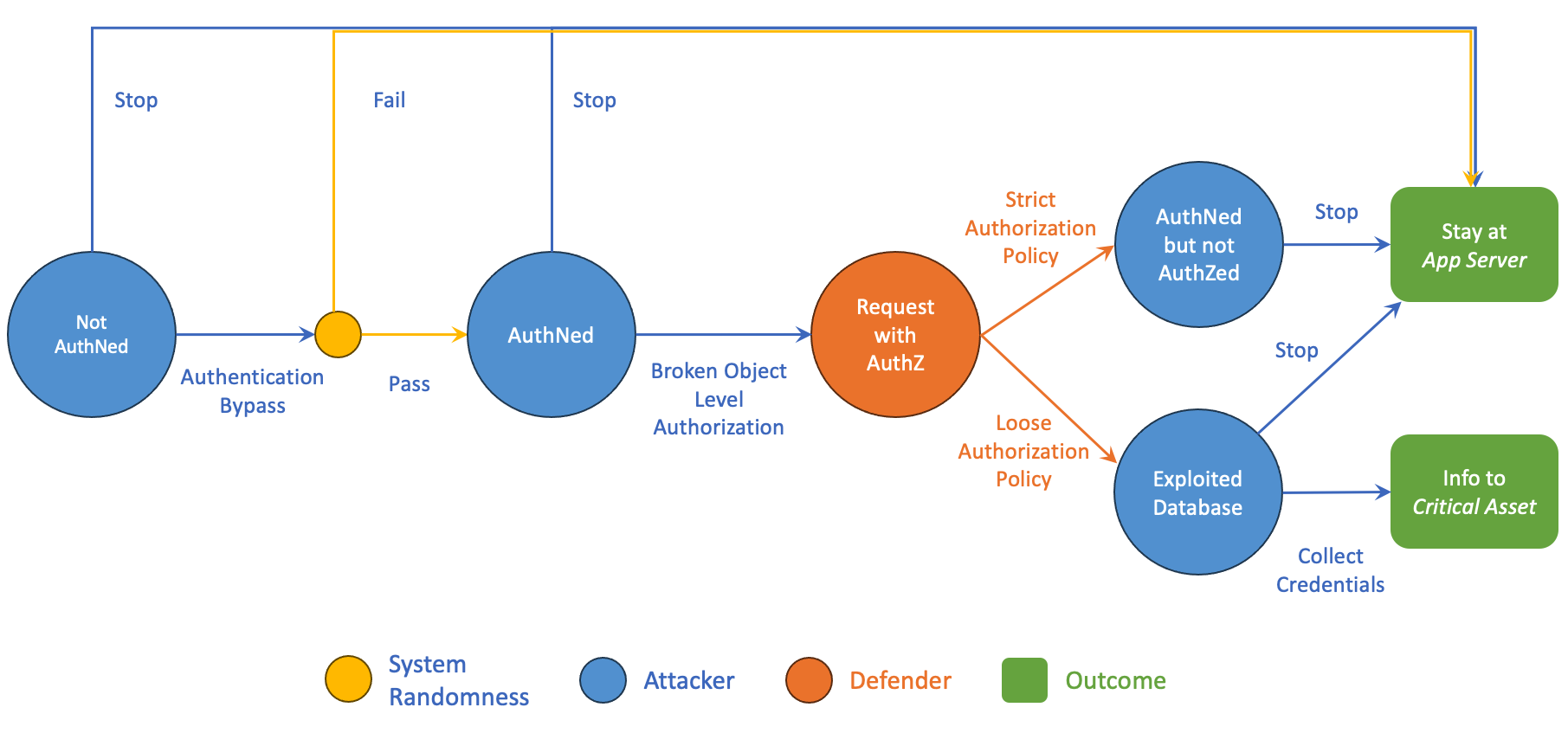}
    \caption{MTG tree at the application server.}\vskip -8mm
    \label{fig:apptree}
    \end{figure}

    
    \begin{figure}[!h]
    \centering
    \includegraphics[width=0.9\linewidth]{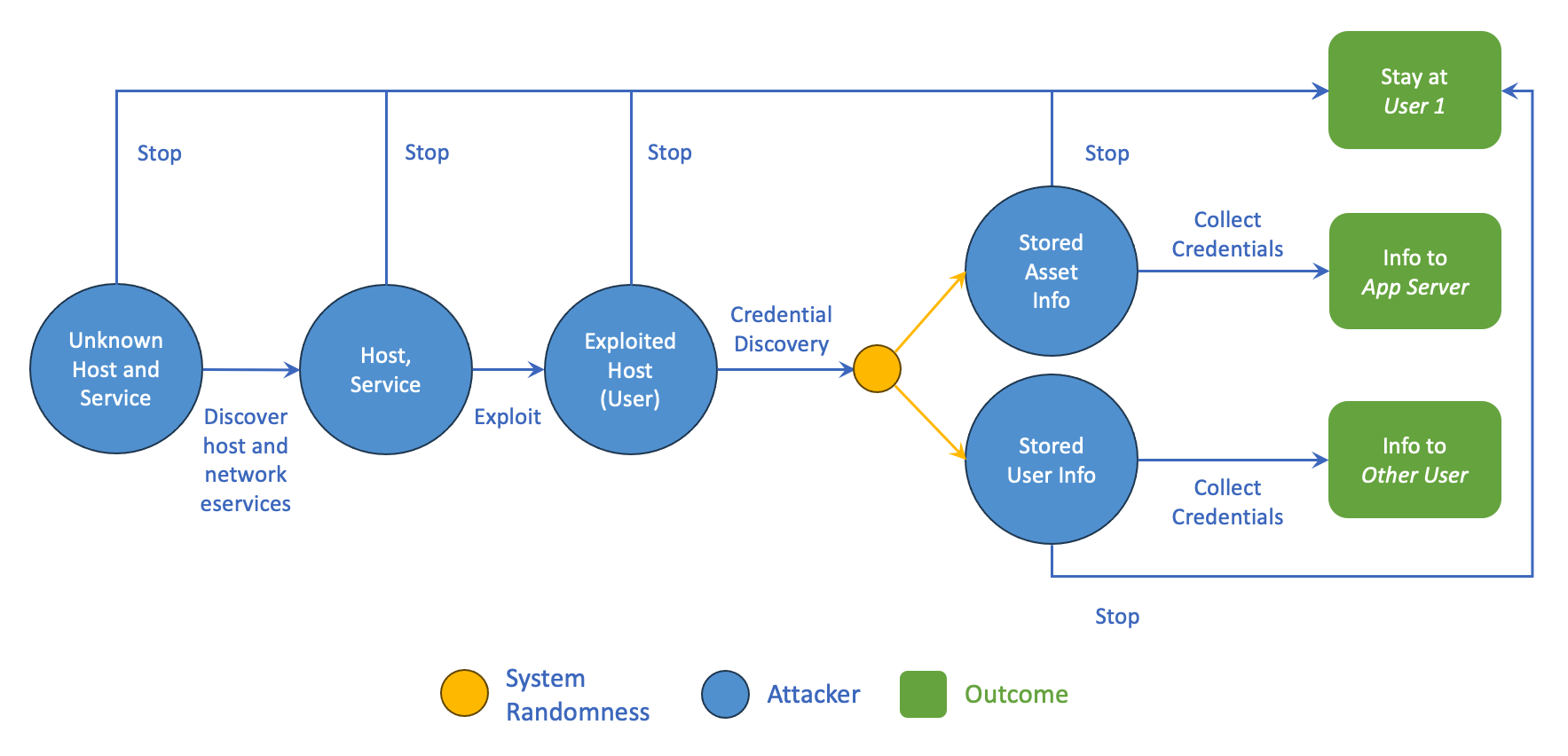}
    \caption{MTG tree at the user device.}\vskip -8mm
    \label{fig:usertree}
    \end{figure}

    
    \begin{figure}[!h]
    \centering
    \includegraphics[width=0.5\linewidth]{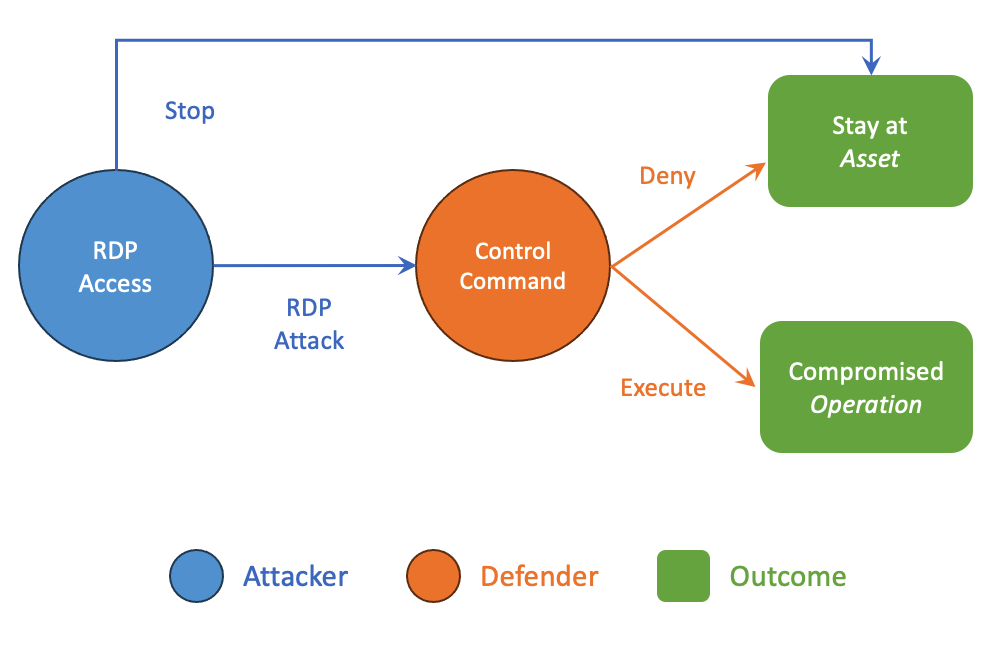}
    \caption{MTG tree at the critical asset.}
    \label{fig:assettree}
    \end{figure}


\end{document}